\documentclass{llncs}

\newcommand{\AlgS}{$\mathcal{A}_1$\xspace}
\newcommand{\AlgLS}{$\mathcal{A}_{2a}$\xspace}
\newcommand{\AlgLL}{$\mathcal{A}_{2b}$\xspace}

\usepackage[utf8]{inputenc}
\usepackage[noend]{algpseudocode}
\usepackage{color}
\usepackage{tikz}
\usepackage{amsfonts,amsmath}
\usepackage{enumerate}
\usetikzlibrary{decorations}
\usetikzlibrary{decorations.pathreplacing}
\usepackage{xspace}
\usepackage{float}
\floatstyle{boxed}
\restylefloat{figure}

\spnewtheorem{fact}[theorem]{Fact}{\bfseries}{\itshape}

\newcommand{\no}[1]{}

\title{Dictionary matching in a stream}
\author{Rapha\"el Clifford\inst{1}
    \and Allyx  Fontaine\inst{1}
    \and Ely Porat\inst{2} 
    \and \\ Benjamin Sach\inst{1}
    \and Tatiana Starikovskaya\inst{1}
}
\institute{University of Bristol, Department of Computer Science, Bristol, U.K.
 \and Bar-Ilan University, Department of Computer Science, Israel}
\date{}

\begin{document}

\maketitle

\begin{abstract}
We consider the problem of dictionary matching in a stream. Given a
set of strings, known as a dictionary, and a stream of characters
arriving one at a time, the task is to report each time some string in
our dictionary occurs in the stream. We present a randomised algorithm
which takes $O(\log{\log(k+m)})$ time per arriving character and uses
$O(k \log{m})$ words of space, where $k$ is the number of strings in
the dictionary and $m$ is the length of the longest string in the
dictionary.
\end{abstract}

\section{Introduction}

We consider the problem of dictionary matching in a stream. Given a
set of strings, known as a dictionary, and a stream of characters
arriving one at a time, the task is to determine when some string in
our dictionary matches a suffix of the growing stream.  The dictionary
matching problem models the common situation where we are interested
in not only a single pattern that may occur but in fact a whole set of
them.

Dictionary matching is considered one of the classic and most widely
studied problems within the field of combinatorial pattern
matching. The original solution of Aho and Corasick~\cite{Aho:1975}
has, for example, been cited over $2800$ times. The dictionary problem
along with its efficient solutions also admit a very wide range of
practical applications: from searching for DNA sequences in genetic
databases~\cite{SE:2005} to intrusion detection~\cite{TSCV:2004} and
many more. The dictionaries that are used in these applications are
often also very large as they may contain all strings within a
neighbourhood of some seed for example, or even all strings in a
language defined by a particular regular expression. As a result,
there is a pressing need for methods which are not only fast but also
use as little space as possible.

The solutions we present will be analysed under a particularly strong
model of space usage. We will account for all the space used by our
algorithm and will not, for example, even allow ourselves to store a
complete version of the input. In particular, we will neither be able
to store the whole of the dictionary nor the streaming text.  We now
define the problem which will be the main object of study for this
paper more formally.

\begin{problem}
In the dictionary-matching problem we have a set of patterns $\mathcal
P$ and a streaming text $T=t_1\dots t_n$ which arrives one character
at a time. We must report all positions in $T$ where there exists a
pattern in $\mathcal P$ which matches exactly. More formally, we
output all the positions $x$ such that there exists a pattern $P_i \in
\mathcal P$ with $t_{x-|P_i|+1}\dots t_{x}=P_i$. We must report an
occurrence of some pattern in $\mathcal P$ as soon as it occurs and
before we can process the subsequent arriving character.
\end{problem}

If all the patterns in the text had the same length $m$ then we could
straightforwardly deploy the fingerprinting method of Karp and Rabin
\cite{KR:1987} to maintain a fingerprint of a window of length $m$
successive characters of the text. We can then compare this for each
new character that arrives to a hash table of stored fingerprints of
the patterns in the dictionary.  In our notation this approach would
require $O(k + m)$ words of space and constant time per
arrival. However if the patterns are not all the same length this
technique no longer works.

For a single pattern, Porat and Porat~\cite{Porat:09} showed that it
is possible to perform exact matching in a stream quickly using very
little space. To do this they introduced a clever combination of the
randomised fingerprinting method of Karp and Rabin and the
deterministic and classical KMP algorithm~\cite{Knuth:1977}.  Their
method uses $O(\log{m})$ words of space and takes $O(\log{m})$ time
per arriving character where $m$ is the length of the single pattern.
Breslauer and Galil subsequently made two improvements to this
method. First, they sped up the method to only require $O(1)$ time per
arriving character and they also showed that it was possible to
eliminate the possibility of false negatives, which could occur using
the previous approach~\cite{BG:2014}.

Our solution takes the single-pattern streaming algorithm of Breslauer
and Galil~\cite{BG:2014} as its starting point.  If we were to run
this algorithm independently in parallel for each separate string in
the dictionary, this would take $O(k)$ time per arriving character and
$O(k\log{m})$ words of space.  Our goal in this paper is to reduce the
running time to as close to constant as possible without increasing
the working space.  Achieving this presents a number of technical
difficulties which we have to overcome.

The first such hurdle is how to process patterns of different lengths
efficiently. In the method of Breslauer and Galil prefixes of power of
two lengths are found until either we encounter a mismatch or a match
is found for a prefix of length at least half of the total pattern
size.  Exact matches for such long prefixes can only occur rarely and
so they can afford to check each one of these potential matches to see
if it can be extended to a full match of the pattern.  However when
the number of patterns is large we can no longer afford to inspect
each pattern every time a new character arrives.

Our solution breaks down the patterns in the dictionary into three
cases: short patterns, long patterns with short periods, long patterns
with long periods.  A key conceptual innovation that we make is a
method to split the patterns into parts in such a way that matches for
all of these parts can be found and stitched together at exactly the
time they are needed. We achieve this while minimising the total
working space and taking only $O(\log{\log(k+m)})$ time per arriving
symbol.

A straightforward counting argument tells us that any randomised
algorithm with inverse polynomial probability of error requires at
least $\Omega(k \log n)$ \emph{bits} of space, see for example
\cite{BroderM:2003}. Our space requirements are therefore within a
logarithmic factor of being optimal. However, unlike the
single-pattern algorithm of Breslauer and Galil, our dictionary
matching algorithm can give both false positives and false negatives
with small probability.

Throughout the rest of this paper, we will refer to the arriving text
character as the arrival.  We can now give our main new result which
will be proven in the remaining parts of this paper.
\begin{theorem}
 Consider a dictionary $\mathcal P$ of $k$ patterns of size at most
 $m$ and a streaming text $T$. The streaming dictionary matching
 problem can be solved in $O(\log{\log(k+m)})$ time per arrival and
 $O(k \log{m})$ words of space.  The probability of error is $O(1/n)$
 where $n$ is the length of the streaming text.  \label{thm:main}
\end{theorem}

\subsection{Related work}

The now standard offline solution for dictionary matching is based on
the Aho-Corasick algorithm \cite{Aho:1975}. Given a dictionary
$\mathcal{P}=\{P_1,P_2,\dots,P_k\}$, and a text $T=t_1 \dots t_n$, let
$\mbox{occ}$ denote the number of matches and $M$ denote the sum of
the lengths of the patterns in $\mathcal{P}$, that is
$M=\sum_{i=1}^{k}|P_i|$. The Aho-Corasick algorithm finds all
occurrences of elements in $\mathcal{P}$ in the text $T$ in
$O(M+n+\mbox{occ})$ time and $O(M)$ space.  Where the dictionary is
large, the space required by the Aho-Corasick approach may however be
excessive.

There is now an extensive literature in the streaming model.  Focusing
narrowly on results related to the streaming algorithm of Porat and
Porat~\cite{Porat:09}, this has included a form of approximate
matching called parameterised matching~\cite{JPS:2013}, efficient
algorithms for detecting periodicity in streams~\cite{EJS:2010} as
well as identifying periodic trends~\cite{CM:2011}. Fast deterministic
streaming algorithms have also been given which provided guaranteed
worst case performance for a number of different approximate pattern
matching problems~\cite{CS:2010,CS:2011} as well as pattern matching
in multiple streams~\cite{CJPS:2012:multiple}.

The streaming dictionary matching problem has also been considered in
a weaker model where the algorithm is allowed to store a complete
read-only copy of the pattern and text but only a constant number of
extra words in working space.  Breslauer, Grossi and Mignosi
\cite{BGM:2011} developed a real-time string matching algorithm in
this model by building on previous work of Crochemore and Perrin
\cite{CP:1991}. The algorithm is based on the computation of periods
and critical factorisations allowing at the same time a forward and a
backward scan of the text.

\subsection{Definitions}

We will make extensive use of Karp-Rabin fingerprints~\cite{KR:1987}
which we now define along with some useful properties.

\begin{definition}{Karp-Rabin fingerprint function $\phi$.}
Let $p$ be a prime and $r$ a random integer in $\mathbb{F}_p$. We
define the fingerprint function $\phi$ for a string $S=s_1 \dots
s_\ell$ such that:
\begin{center}
 $\phi(S)=\sum_{i=1}^{\ell} s_ir^i ~mod~ p$.
\end{center}
\end{definition}

The most important property is that for any two equal length strings
$U$ and $V$ with $U \neq V$, the probability that $\phi(U) = \phi(V)$
is at most $1/n^2$ if $p>n^3$. We will also exploit several well known
arithmetic properties of Karp-Rabin fingerprints which we give in
Lemma~\ref{lem:addsubphi}. All operations will be performed in the
word-RAM model with word size $\Theta(\log{n})$.

\begin{lemma}\label{lem:addsubphi}
Let $U$ be a string of size $\ell$ and $V$ another string, then:
\begin{itemize}
\item $\phi(UV) = \phi(U) + r^{\ell} \phi(V)~ mod~p$,
\item $\phi(U) =  \phi(UV) - r^{\ell} \phi(V)~ mod~p$,
\item $\phi(V) =  r^{-\ell}(\phi(UV) - \phi(U))~ mod~p$.
\end{itemize}
\end{lemma}

For a non-empty string $x$, an integer $p$ with $0 < p \leq |x|$ is
called a \emph{period} of $x$ if $x_i = x_{i+p}$ for all $i \in
\{1,\dots, |x|-p-1\}$. The period of a non-empty string $x$ is simply
the smallest of its periods.  We will also assume that all logarithms
are base 2 and are rounded to the nearest integer.

We describe three algorithms: $\mathcal{A}_1$ in Section
\ref{sec:shortpatterns} which handles short patterns in the
dictionary, and $\mathcal{A}_{2a}$ and $\mathcal{A}_{2b}$ in Section
\ref{sec:longpatterns} which deal with the long
patterns. Theorem~\ref{thm:main} is obtained by running all three
algorithms simultaneously.

\section{Algorithm $\mathcal{A}_{1}$: Short patterns}\label{sec:shortpatterns}
\begin{lemma}
\label{lm:alg3}
There exists an algorithm $\mathcal{A}_1$ which solves the streaming
dictionary matching problem and runs in $O(\log \log (k+m))$ time per
arrival and uses $O(k \log m)$ space on a dictionary of $k$ patterns
whose maximum length is at most $2k\log{m}$.
\end{lemma}

For very short patterns shorter than $2 \log m$ we can
straightforwardly construct an Aho-Corasick
automaton~\cite{Aho:1975}. To make this efficient we store a static
perfect hash table at each node to navigate the automaton. The
automaton occupies at most $O(k \log m)$ space and reports occurrences
of short patterns in constant time per arrival. From now on, we can
assume that all patterns are longer than $2 \log m$.

Our solution splits each of the patterns, which are all now guaranteed
to have length greater than $2\log{m}$, into two parts in multiple
ways. The first part of each splitting of the pattern we call the head
and the rest we call the tail.  Tails will always have length $\ell$
for all $\ell$ s.t. $\log m < \ell \le 2 \log m$.  We will therefore
split each pattern into at most $\log{m}$ head/tail pairs, making a
total of at most $k\log{m}$ heads overall.

The overall idea is to insert all heads into a data structure so that
we can find potential matches in the stream efficiently. We will only
look for potential matches every $\log{m}$ arrivals. We use the
remaining at least $\log{m}$ arrivals before a full match can occur
both to de-amortise the cost of finding head matches as well as to
check if the relevant tails match as well.

In order to look for matches with heads of patterns efficiently we
will use a slight modification of the probabilistic $z$-fast trie
introduced by Djamal Bellazougui et al.~\cite{DBPV:2009} (Theorem
4.1~\cite{DBPV:2009}). A $z$-fast trie is a randomised data structure
which compactly represents a trie on a set of strings. Our
modification to the probabilistic $z$-fast trie simply uses a
different signature function. For a string $s = s_1\ldots s_k$ we
define it to be $\phi(s_k\ldots s_1)$, the fingerprint of the reverse
of $s$. Otherwise the data structure remains unchanged.

An important concept in this data structure is the \emph{exit node} of
a string $x$. This is the deepest node labelled by a prefix of
$x$. Given a string $x$ and signatures of all its prefixes, we can
find the exit node of $x$ using the $z$-fast trie in $O(\log m + \log
\log (k+m))$ time, where $m$ is the maximal length of the
strings. Importantly, the lookup procedure compares at most $\log m +
\log \log (k+m)$ pairs of signatures, and hence the probability of a
false match is at most $\frac{\log m + \log \log (k+m)}{n^2} <
\frac{1}{n}$. When there are no false positives in signatures
comparison, correctness and the time bound are guaranteed by Lemma~4.2
and Lemma~4.3 of~\cite{DBPV:2009}.

We can now describe Algorithm $\mathcal{A}_1$ assuming that all
patterns are longer than $2 \log m$ but no longer than $2k\log{m}$. As
a preprocessing step, we build the probabilistic $z$-fast trie for the
reverse of the at most $k \log m$ heads. For regularly spaced indices
of the text, we will use the $z$-fast trie to find the longest head
that matches at each of these locations.

We will also augment the $z$-fast trie in the following way. We mark
each node labelled by a head with a colour representing the
fingerprint of the corresponding tail. In the end, each node may be
marked by several colours, and the total number of colours will be $k
\log m$. On top of the $z$-fast trie we build coloured-ancestor data
structure~\cite{MM:1996}. This occupies $O(k \log m)$ space and
supports $\mathsf{Find} (u, c)$ queries in $O (\log \log \bigl (k \log
m \bigr)) = O(\log \log (k+m))$ time, where $\mathsf{Find} (u, c)$ is
the lowest ancestor of a node $u$ marked with colour $c$. Each pattern
consists of one head concatenated to its corresponding tail and so we
will use coloured-ancestor queries to find the longest whole pattern
matches by using the fingerprints of different tails as queries.

At all times we maintain a circular buffer of size $2k \log m$ which
holds the fingerprints of the at most recent $2k \log m$ prefixes of
the text. Let $i$ be an integer multiple of $\log m$. For each such
$i$, we query the $z$-fast trie with a string $x = t_{i} \ldots t_{i -
  2k \log m + 1}$. Note that for each prefix of $x$ we can compute its
signature in $O(1)$ time with the help of the buffer. The query
returns the exit node $e(x)$ of $x$ in $O(\log m + \log \log (k+m))$
time, which is used to analyse arrivals in the interval $[i + \log m,
  i + 2\log m]$. This exit node corresponds to the longest head that
matches ending at index $i$. The $O(\log m)$ cost of performing the
query is de-amortised during the interval $(i, i + \log m]$.

For each arrival $t_\ell$, $\ell \in (i + \log m, i + 2\log m]$ we
  compute the fingerprint $\phi$ of $t_{i +1} \ldots t_\ell$. This can
  be done in constant time as we store the last $2k \log m \ge m >
  2\log m$ fingerprints. If $\mathsf{Find} (e (x), \phi)$ is defined,
  $\ell$ is an endpoint of a whole pattern match and we report
  it. Otherwise, we proceed to the next arrival. The overall time per
  arrival is therefore dominated by the time to perform the
  coloured-ancestor queries which is $O(\log \log (k+m))$ .

We remark that the algorithm can be applied also to patterns of
maximal length $4k \log m$ and the time complexity will be
unchanged. Moreover, if there are several possible patterns that match
for a given arrival, the algorithm reports the longest such
pattern. These two properties will be needed when we describe
Algorithm $\mathcal{A}_{2b}$ in Section \ref{sec:longper}.

\section{Long patterns}\label{sec:longpatterns}
We now assume that all the patterns have length greater than
$2k\log{m}$. We distinguish two cases according to the periodicity of
those patterns: those with short period and those with long period.
Hereafter, to distinguish the cases, we use the following
notation. Let $m_i = |P_i|$ and $Q_i$ be the prefix of $P_i$ such that
$|Q_i|=m_i-k\log{m}$.  Let $\rho_{Q_i}$ be the period of $Q_i$. The
remaining patterns are then partitioned in two disjoint groups of
patterns, those with $\rho_{Q_i} < k\log{m}$ and those with
$\rho_{Q_i} \geq k\log{m}$. We describe two algorithms:
$\mathcal{A}_{2a}$ and $\mathcal{A}_{2b}$, one for each case
respectively.  Finally, the overall solution is then to run all three
algorithms $\mathcal{A}_{1}$, $\mathcal{A}_{2a}$, $\mathcal{A}_{2b}$
simultaneously to obtain Theorem~\ref{thm:main}.

	\subsection{Algorithm $\mathcal{A}_{2a}$: Patterns with short periods}\label{sec:shortper}
	This section gives an algorithm for a dictionary of patterns $\mathcal
P=P_1,\dots,P_k$ such that $m_i \geq 2k \log m$ and $\rho_{Q_i} <
k \log m$. Recall that $Q_i$ is the prefix of $P_i$ of length $m_i -
k \log m$ and $\rho_{Q_i}$ is the period of $Q_i$. The overall idea
for this case is that if we can find enough repeated occurrences of
the period of a pattern then we know we have almost found a full
pattern match. As the pattern may end with a partial copy of its
period we will have to handle this part separately. The main technical
hurdle we overcome is how to process different patterns with different
length periods in an efficient manner.

We define the tail of a pattern $P_i$ to be its suffix of length
$2k \log m$. Observe that a $P_i$ match occurs if and only if there is
a match of $Q_i$ followed by a match with the tail of $P_i$.

Let $K_i$ be the prefix of $Q_i$ of length $k \log m$. Further observe
that $Q_i$ can only match if there is a sequence of
$\left\lfloor \frac{|Q_i|-|K_i|}{\rho_{Q_i}}+1 \right\rfloor$
occurrences of $K_i$ in the text, each occurring exactly $\rho_{Q_i}$
characters after the last. This follows immediately from the fact that
$K_i$ has length $k \log m$ and $Q_i$ has period $\rho_{Q_i} < k \log
m$.

We now describe algorithm \AlgLS which solves this case. At all times
we maintain a circular buffer of size $2k \log m$ which holds the
fingerprints of the most recent $2k \log m$ prefixes of the text. That
is, if the last arrival is $t_{\ell}$, then the buffer contains the
fingerprints $\phi(t_1 \dots t_{\ell - 2k\log m + 1})$, $\dots$,
$\phi(t_1 \dots t_{\ell})$.

To find $K_i$ matches, we store the fingerprint $\phi(K_i)$ of each
distinct $K_i$ in a static perfect hash table. By looking up
$\phi(t_{\ell-k\log m+1} \dots t_{\ell})$ we can find whether some
$K_i$ matches in $O(1)$ time. For each distinct $K_i$ we maintain a
list of recent matches stored as an arithmetic progression. Each time
we find a new match with $K_i$ we check whether it is exactly
$\rho_{Q_i}$ characters from the last match. If so we include it in
the current arithmetic progression. If not, then we delete the current
progression and start a new progression containing only the latest
match. Note that $K_i=K_j$ implies that $\rho_{Q_i} = \rho_{Q_j}$ and
therefore there is no ambiguity in the description.

We store the fingerprint of each tail in another static perfect hash
table. For each arrival $t_\ell$ we use this hash table to check
whether $\phi(t_{\ell-2k\log m+1} \dots t_{\ell})$ matches the
fingerprint of some tail. This takes $O(1)$ time per arrival.

Assume that the tail of some $P_i$ matched. We will justify below that
we can assume that each tail corresponds to a unique $P_i$. It remains
to decide whether this is in-fact a full match with $P_i$. This is
determined by a simple check, that is whether the current arithmetic
progression for $K_i$ contains at least
$\left\lfloor \frac{|Q_i|-|K_i|}{\rho_{Q_i}}+1 \right\rfloor$
occurrences.
\begin{lemma}
\label{lem:A2a}
Algorithm $\mathcal{A}_{2a}$ takes $O(1)$ time per character and uses
$O(k \log m)$ space.
\end{lemma}
\begin{proof}

The algorithm stores two hash tables, each containing $O(k\log m)$
fingerprints as well as $O(k)$ arithmetic progressions. The total
space is therefore $O(k \log m)$ as claimed. The time complexity of
$O(1)$ per character follows by the use of static perfect hash tables
(which are precomputed and depend only on~$\mathcal P$).

We first prove the claim that each tail corresponds to a unique
$P_i$. To this end, we assume in this section that no pattern contains
another pattern as a suffix. In particular, any such pattern can be
deleted from the dictionary during the preprocessing stage as it does
not change the output. This implies the claim that each $P_i$ has a
distinct tail because the tail contains a full period of $P_i$.

The correctness follows almost immediately from the algorithm
description via the observation that each $Q_i$ is formed from
$\left\lfloor \frac{|Q_i|-|K_i|}{\rho_{Q_i}}+1 \right\rfloor$ repeats
of $K_i$ followed by a prefix of $K_i$. We check explicitly whether
there are sufficient repeats of $K_i$ in the text stream to imply a
$Q_i$ match. While we do not check explicitly that either final prefix
of $K_i$ is a match or that the full $P_i$ matches, this is implied by
the tail match. This is because the tail has length $2k \log m$ and
hence includes the final prefix of $K_i$ and the last $k \log m$
characters of $P_i$ (those in $P_i$ but not in~$Q_i$).
\qed
\end{proof}

	\subsection{Algorithm $\mathcal{A}_{2b}$: Patterns with long periods}\label{sec:longper}
	
Consider a dictionary $\mathcal P$ in which the patterns are such that
$m_i \geq 2k \log m$ and $\rho_{Q_i} \geq k \log m$. Let us define $k$
to be number strings in this dictionary. We can now describe
Algorithm \AlgLL.  Recall that $Q_i$ is the prefix of $P_i$
s.t. $|Q_i|=m_i-k\log m$. For each pattern $P_i$, we define $P_{i,j}$
to be the prefix of $P_i$ with length $2^j$, $1 \le 2^j \le
m_i-2k \log m$.

We will first give an overview of an algorithm that identifies
$P_{i,j}$ matches in $O(\log m)$ time per arrival. With the help
of \AlgS and \AlgLS we will speed it up to achieve an algorithm with
$O(\log \log (k+m))$ time per arrival. The algorithm will identify the
matches with a small delay up to $k \log m$ arrivals. We then show how
to extend $P_{i,j}$ to $Q_i$ matches. This stage will still report the
matches after they occur. Finally we show how to find whole pattern
matches in the stream using the $Q_i$ matches while also completely
eliminating the delay in the reporting of these machines.  In other
words, any matches for whole patterns will be reported as soon as they
occur and before the next arrival in the stream as desired.

\subsubsection{$O(\log m)$-time algorithm.}
We define a logarithmic number of \emph{levels}. Level $j$ will
represent all the matches for prefixes $P_{i,j}$. We store
only \textit{active} prefix matches, that still have the potential to
indicate the start of full matches of a pattern in the
dictionary. This means that any match at level $j$ whose position is
more than $2^{j+1}$ from the current position of an arrival is simply
removed. We will use the following well-known fact.

\begin{fact}[Lemma 3.2\cite{BG:2014}]\label{lem:progressions}
If there are at least three matches of a string $U$ of length $2^{j}$
in a string $V$ of length $2^{j+1}$, then positions of all matches of
$U$ in $V$ form an \textit{arithmetic progression}. The difference of
the progression is equal to the length of the period of $U$.
\end{fact}

It follows that if there are at least three active matches for the
same prefix at the same level, we can compactly store them as
a \textit{progression} in constant space. Consider a set of distinct
prefixes of length $2^{j}$ of the patterns in $\mathcal P$. For each
of them we store a progression that contains:

\begin{enumerate}[{(1)}]
	\item The position $\texttt{fp}$ of the first match;
	\item The fingerprint of $t_1 \ldots t_{\texttt{fp}}$;
	\item The fingerprint of the period $\rho$ of the prefix;
	\item The length of the period $\rho$ of the prefix;
	\item The position $\texttt{lp}$ of the last match.
\end{enumerate}

With this information, we can deduce the position and the fingerprint
of the text from the start to the position of any active match of the
prefix. Moreover, we can add a new match or delete the first match in
a progression in $O(1)$ time.

We make use of a perfect hash table $\mathcal{H}$ that stores the
fingerprints of all the prefixes of the patterns in $\mathcal P$. The
keys of $\mathcal{H}$ correspond to the fingerprints of all the
prefixes and the associated value indicates whether the prefix from
which the key was obtained is a proper prefix of some pattern, a whole
pattern itself, or both. Using the construction of~\cite{Ruzic:08},
for example, the total space needed to store all the fingerprints and
their corresponding values is $O(k\log{m})$.

When a character $t_\ell$ of the text arrives, we update the current
position and the fingerprint of the current text. The algorithm then
proceeds by the progressions over $\log m$ levels.  We start at level
$0$. If the fingerprint $\phi(t_\ell)$ is in $\mathcal{H}$, we insert
a new match to the corresponding progression at level $0$.

For each level $j$ from $0$ to $\log m$, we retrieve the position $p$
of the first match at level $j$. If $p$ is at distance $2^{j+1}$ from
$t_\ell$, we delete the match and check if the fingerprint
$\phi(t_p \ldots t_\ell)$ is in $\mathcal{H}$. If it is and the
fingerprint is a fingerprint of one of the patterns, we report a match
(ending at $t_\ell$, the current position of the text). If the
fingerprint is in $\mathcal{H}$ and if it is a fingerprint of a proper
prefix, then $p$ is a plausible position of a match of a prefix of
length $2^{j+1}$. We check if it fits in the appropriate progression
$\pi$ at level $j+1$. (Which might not be true if the fingerprints
collided). If it does, we insert $p$ to $\pi$. If $p$ does not match
in $\pi$, we discard it and proceed to the next level.

As updating progressions at each level takes $O(1)$ time only, and
there are $\log m$ levels, the time complexity of the algorithm is
$O(\log m)$ per arrival. The space complexity is $O(k \log m)$.  We
deliberately omit some details (for example, how to retrieve the
position of the first match in the level) as they will not be
important for the final algorithm.

\subsubsection{$O(\log\log (k+m))$-time algorithm.}
We will follow the same level-based idea. To speed up the algorithm,
we will consider prefixes $P_{i,j}$ with short and long periods
separately. The number of matches of the prefixes with short periods
can be big, but we will be able to compute them fast with the help
of \AlgS and \AlgLS. On the other hand, matches of the prefixes with
long periods are rare, and we will be able to compute them in a round
robin fashion.

Let $\rho_{i,j}$ be the period of $P_{i,j}$. We first build a
dictionary $D_1$ containing at most one prefix for each
$P_{i}$. Specifically, containing the largest $P_{i,j}$ with the
period $\rho_{i,j}<k \log m$ and $2k \log m \leq |P_{i,j}| \le
m_i-2k \log m$. If no such $P_{i,j}$ exists we do not insert a prefix
for $P_i$. This dictionary is processed using a modification of
algorithm \AlgLS which we described in Section~\ref{sec:shortper}. The
modification is that when a text character $t_\ell$ arrives, the
output of the algorithm identifies the longest pattern in $D_1$ which
matches ending at $t_\ell$ or `no match' if no pattern matches. This
is in contrast to \AlgLS as described previously where we only
outputted whether some pattern matches. The modification takes
advantage of the fact that prefixes in $D_1$ all have power-of-two
lengths and uses a simple binary search approach over the $O(\log m)$
distinct pattern lengths. This increases the run-time of \AlgLS to
$O(\log \log m)$ time per arrival. The details can be found in
Appendix~\ref{sec:app}.

Whenever a match is found with some pattern in $D_1$, we update the
match progression of the reported pattern (but not of any of its
suffixes that might be in $D_1$). Importantly, we will still have at
most two progressions of active matches per prefix because of the
following lemma and corollary.

\begin{lemma}\label{lem:equalperiods}
Let $P_{i,j}, P_{i',j'}$ be two prefixes in $D_1$ and suppose that
$P_{i,j}$ is a suffix of $P_{i',j'}$. The periods of $P_{i,j},
P_{i',j'}$ are equal.
\end{lemma}
\begin{proof}
Assume the contrary. Then $P_{i,j}$ has two periods: $\rho_{i,j}$ and
$\rho_{i',j'}$ (because it is a suffix of $P_{i',j'}$). We have
$\rho_{i,j} + \rho_{i',j'} < 2k \log m \le |P_{i,j}|$. By the
periodicity lemma (see, e.g.,~\cite{Lothaire:2002}), $\rho_{i,j}$ is a
multiple of $\rho_{i',j'}$. But then $P_{i,j}$ is periodic with period
$\rho_{i',j'} < \rho_{i,j}$, a contradiction.
\qed
\end{proof}

\begin{corollary}\label{cor:equalperiods}
Let $P_{i,j}$, $P_{i',j'}$, and $P_{i'',j''}$ be prefixes in
$D_1$. Suppose that $P_{i,j}$ is a suffix of $P_{i',j'}$ and
simultaneously is a suffix of $P_{i'',j''}$. Then $P_{i',j'}$ is a
suffix of $P_{i'',j''}$ (or vice versa).
\end{corollary}

We now consider any $P_i$ for which we did not find a suitable small
period prefix. In this case it is guaranteed that there is a prefix
$P_{i,j}$ with the period longer than $k \log m$ but length at most
$4k \log m$. We build another dictionary $D_2$ for each of these
prefixes. We apply algorithm \AlgS and for each arrival $t_\ell$
return the longest prefix $P_{i,j}$ in $D_2$ that matches at it in
$O(\log \log (k+m))$ time. We then need to update the match
progression of $P_{i,j}$ as well as the match progressions of all
$P_{i',j'} \in D_2$ that are suffixes of $P_{i,j}$. Fortunately, each
of the prefixes in $D_2$ can match at most once in every $k \log m$
arrivals, because the period of each of them is long, meaning that we
can schedule the updates in a round robin fashion to take $O(1)$ time
per arrival.

We denote a set of all $P_{i,j}$ such that $\rho_{i,j} \ge k \log m$
by $S$. Any of these prefixes can have at most one match in $k \log m$
arrivals. Because of that and because $|S| \leq k \log m$, we will be
able to afford to update the matches in a round robin fashion.

We will have two update processes running in parallel. The first
process will be updating matches of prefixes $P_{i,j} \in S$ such that
$P_{i,j-1} \in S \cup D_2$. We consider one of these prefixes per
arrival. If there is a match with $P_{i,j}$ in $[t_\ell-k\log m,
t_{\ell}]$ then there must be a corresponding match with $P_{i,j-1}$
ending in $[t_{\ell-2^{j-1}-k\log m}, t_{\ell-{2^{j-1}}}]$. As
$P_{i,j-1} \in S$, $\rho_{i,j} \ge k\log m$ so there is at most one
match. We can determine whether this match can be extended into a
$P_{i,j}$ match using a single fingerprint comparison as described in
the $O(\log m)$-time algorithm. This is facilitated by storing a
circular buffer of the fingerprints of the most recent $k\log m$ text
prefixes.

The second process will be updating matches of prefixes $P_{i,j} \in
S$ such that $P_{i,j-1} \in D_1$. Again, if there is a match with
$P_{i,j}$ in $[t_\ell-k\log m, t_{\ell}]$ then there must be a
corresponding match with $P_{i,j-1}$ ending in $[t_{\ell-2^{j-1}-k\log
m}, t_{\ell-{2^{j-1}}}]$. However, the second process will be more
complicated for two reasons. First, $P_{i,j-1}$ has a small period so
there could be many $P_{i,j-1}$ matches ending in this
interval. Second, the information about $P_{i,j-1}$ matches can be
stored not only in the progressions corresponding to $P_{i,j-1}$, but
also in the progressions corresponding to prefixes that have
$P_{i,j-1}$ as a suffix. The first difficulty can be overcome because
of the following lemma.

\begin{lemma}\label{lem:changeofperiod}
Consider any $P_{i,j}$ such that $\rho_{i,j-1} < k \log
m \leq \rho_{i,j}$. Given a match progression for $P_{i,j-1}$, only
one match could also correspond to a match with $P_{i,j}$.
\end{lemma}
\begin{proof}
Let $U$ be the prefix of $P_{i,j-1}$ of length $\rho_{i,j-1}$. That
is, the substrings bounded by consecutive matches in the match
progression for $P_{i,j-1}$ are equal to $U$. Suppose that $P_{i,j}$
starts with exactly $r$ copies of $U$. Then we have $P_{i,j} = U^r V$
for some string $V$. Note that as $\rho_{i,j-1} < k \log
m \leq \rho_{i,j}$, the string $V$ cannot be a prefix of $U$. Then the
only match in the progression which could match with $P_{i,j}$ is the
$r$-th last one.
\qed
\end{proof}

To overcome the second difficulty, we use
Corollary~\ref{cor:equalperiods}. It implies that prefixes in $D_1$
can be organized in chains based on the ``being-a-suffix"
relationship. We consider prefixes in each chain in a round robin
fashion again. We start at the longest prefix, let it be $P_{i,j}$. At
each moment we store exactly one progression initialized to the
progression of $P_{i,j}$. If the progression intersects with
$[t_{\ell-2^{j-1}-k\log m}, t_{\ell-{2^{j-1}}}]$, we identify the
`interesting' match in $O(1)$ time with the help of
Lemma~\ref{lem:changeofperiod} and try to extend it as in the first
process. We then proceed to the second longest prefix $P_{i',j'}$. If
the stored progression intersects with $[t_{\ell-2^{j'-1}-k\log m},
t_{\ell-2^{j'-1}}]$, we proceed as for $P_{i,j}$. Otherwise, we update
the progression to be the progression of $P_{i',j'}$ and repeat the
previous steps for it. We continue this process for all prefixes in
the chain.

From the description of the processes it follows that the matches for
each $P_{i,j}$ (in particular, for the longest $P_{i,j}$ for each $i$)
are outputted in $O(\log \log (k+m))$ time per arrival with a delay of
up to $k \log m$ characters (i.e. at most $k \log m$ characters after
they occur).

\subsubsection{Finding $Q_i$ matches.}
We now show how to find $Q_i$ matches using $P_{i,j}$ matches. If
there is a match with $Q_i$ in $[t_\ell-k\log m, t_{\ell}]$, there
must be a match with the longest $P_{i,j}$ in $[t_\ell-2^j-k\log m,
t_{\ell}-2^j]$. Because $|P_{i,j}| \le m_i - 2k\log m$, this match has
been identified by the algorithm and it is the first match in the
progressions. We can determine whether this match can be extended into
a $Q_i$ match using a single fingerprint comparison.

Therefore the $Q_i$ matches are outputted in $O(\log \log (k+m))$ time
with a delay of up to $k \log m$ characters (i.e. at most $k \log m$
characters after they occur). We can then remove this delay using
coloured ancestor queries in a similar manner to algorithm \AlgS as
described below.

\subsubsection{Finding whole pattern matches and removing the delay.}
Up to this point, we have shown that we can find each $Q_i$ match in
$O(\log \log (k+m))$ time per arrival with a delay of at most $k\log
m$ characters. Further we only report one $Q_i$ match at each time. We
will show how to extend these $Q_i$ matches into $P_i$ matches using
coloured ancestor queries in $O(\log\log (k+m))$ time per arrival.

Build a compacted trie of the reverse of each string $Q_i$. The edges
labels are not stored.  The space used is $O(k)$.  For each $i$ we can
find the reverse of $Q_i$ in the trie in $O(1)$ time (by storing an
$O(k)$ space look-up table).

The tail of each $P_i$ is its $(k\log m)$-length suffix, i.e.\@ the
portion of $P_i$ which is not in $Q_i$. Each distinct tail is
associated with a colour. As there are at most $k \log m$ patterns,
there are at most $k\log m$ colours. Computing the colour from the
tail is achieved using a standard combination of fingerprinting and
static perfect hashing. For each node in the tree which represents
some $Q_i$ we colour the node with the colour of the tail of $P_i$.

Whenever we find a $Q_i$ match, we identify the place in the tree
where the reverse of $Q_i$ occurs. Recall that these matches may be
found after a delay of at most $k \log m$ characters. A $Q_i$ match
ending at position $\ell-k\log m$ implies a possible $P_i$ match at
position $\ell$. We remember this potential match until $t_\ell$
arrives.

More specifically when $t_\ell$ arrives we determine the node $u$ in
the trie representing the reverse of the longest $Q_i$ which has a
match at position $\ell-k\log m$. This can be done in $O(1)$ time by
storing a circular buffer of fingerprints.

We now need to decide whether $Q_i$ implies the existence of some
$P_j$ match.  It is important to observe that as we discarded all but
the longest such $Q_i$, we might find a $P_j$ with $j\neq i$.

For each arrival $t_\ell$, we compute the fingerprint $\phi$ of
$t_{\ell-k\log m+1} \ldots t_\ell$. This can be done in constant time
as we store the last $k\log m$ fingerprints. If $\mathsf{Find}
(u, \phi)$ is defined, $t_\ell$ is an endpoint of a pattern match and
we report it. Otherwise, we proceed to the next arrival.

\begin{lemma}
\label{lm:alg2b}
Algorithm $\mathcal{A}_{2b}$ takes $O(\log \log (k+m))$ time per
character. The space complexity of the algorithm is $O(k \log m)$.
\end{lemma}

\bibliographystyle{splncs03}
\bibliography{longnames,dictionary}
\newpage
\appendix

\section {Suffixes, powers-of-two and the longest match}\label{sec:app} 

In Section~\ref{sec:longper} we will use algorithm \AlgLS as a black
box. However, we will need the output to determine the longest pattern
that matches when each new text character arrives rather than simply
whether a pattern matches. Furthermore, we will not be able to
guarantee (as is safely assumed above) that no pattern is a prefix of
another. Fortunately the patterns will all have a power-of-two
length. We now briefly describe the required changes which increase
the running time from $O(1)$ to $O(\log \log m)$.

The changes do not affect the algorithm until the point at which some
tail, has been matched. As one pattern could be a suffix of another,
$\Theta(\log m)$ patterns could have the same tail. This follows from
the fact that the tail contains a full period of any pattern $P_i$ and
that all patterns have power-of-two lengths.

Whenever a tail is matched when some $t_\ell$ arrives, we need to
determine the longest matching $P_i$ with this tail. Assume, as a
motivating special case, that every $P_i$ with this tail has the same
$K_i$. As above, $P_i$ is associated with a number of occurrences, \[
c_i = \left\lfloor \frac{|Q_i|-|K_i|}{\rho_{Q_i}}+1 \right\rfloor \]
of $K_i$ that are required for a $P_i$ match. The basic idea is to
perform binary search on the set of $c_i$ values (for $P_i$s with the
matching tail) using the number of occurrences of $K_i$ in the current
arithmetic progression as the key. As there are most $O(\log m)$
candidates, this takes $O(\log \log m)$ time.

However, two patterns $P_i$ and $P_j$ with the same tail could have
$K_i \neq K_j$. Fortunately, Lemma~\ref{lem:wrongK} below says that
using the `wrong' $K_i$ only affects the number of required matches by
at most $1$. For each tail, we (arbitrarily) preselect a single $K_i$
among the $P_i$ with this tail. We then perform the same binary search
using $K_i$. As the $O(\log m)$ candidates have power-of-two length
(greater than $2k\log m$) for any two patterns $P_i \neq P_j$, we have
that $|c_i-c_j| > 4$. Therefore, we find at most one candidate, $P_j$
is checked using its own $K_j$.

\begin{lemma}\label{lem:wrongK}
Let $P_i$ and $P_j$ be two patterns with the same tail but $K_i \neq
K_j$. Let us also assume that the tail of $P_j$ matches when some
$t_\ell$ arrives. $P_i$ matches ending at $t_\ell$ if the current
arithmetic progression for $K_j$ contains at least $c_i + 1$
occurrences. Furthermore $P_i$ \emph{does not} match at $t_\ell$ if
the same progression contains fewer than $c_i-1$ matches.
\end{lemma}

\begin{proof}
Let $y_i$ be the number of matches of $K_i$ in the current
progression. Analogously, for $y_j$. The first thing to observe is
that $|y_i-y_j|<1$. This follows from the fact that $|K_i|=|K_j|$,
they are both periodic and contain each other's period string.

Assume that $y_j < c_i-1$. Therefore, as $c_i<c_j+1$, we have that
$y_j < c_i$ so $P_i$ does not match.  Instead assume that $y_j \geq
c_i+1$. Again, as $c_i>c_j-1$, we have that $y_j\geq c_i$. \qed
\end{proof}

\end{document}